\documentclass[11pt]{article}

\usepackage{graphicx}
\usepackage{amssymb}
\usepackage{amsthm}
\usepackage{geometry}
\usepackage{amsmath}
\usepackage{color}  
\usepackage{pstricks}
\usepackage{pst-node}
\usepackage{pst-plot}
\usepackage{epstopdf}
\usepackage{caption}
\DeclareGraphicsExtensions{.pdf,.png,.jpg}
\makeatletter
\newtheorem*{rep@theorem}{\rep@title}
\newcommand{\newreptheorem}[2]{%
\newenvironment{rep#1}[1]{%
 \def\rep@title{#2 \ref{##1}}%
 \begin{rep@theorem}}%
 {\end{rep@theorem}}}
\makeatother

\theoremstyle{plain}
\newtheorem{theorem}{Theorem}
\newreptheorem{theorem}{Theorem}
\newtheorem{proposition}{Proposition}
\newreptheorem{proposition}{Proposition}
\newtheorem{lemma}{Lemma}
\newreptheorem{lemma}{Lemma}
\newtheorem{corollary}{Corollary}
\newreptheorem{corollary}{Corollary}

\newtheorem{fact}{Fact}
\newreptheorem{fact}{Fact}

\newcommand{\etal}{\emph{et al. }}
\newcommand{\ie}{i.e., }   
\newcommand{\eg}{e.g., }  
\newcommand{\cf}{cf. }    
\newcommand{\amax}{\operatorname{argmax}} 
\newcommand{\amin}{\operatorname{argmin}} 
\newcommand{\nir}{\operatorname{NIR}} 
\newcommand{\ir}{\operatorname{IR}} 

\title{Inequality and Network Formation Games}

\author{	Samuel D. Johnson \\
		 Department of Computer Science \\
		 University of California, Davis \\
		 \texttt{samjohnson@ucdavis.edu}
		 \and
		 Raissa M. D'Souza \\
		Department of Computer Science \\
		University of California, Davis \\
		\texttt{raissa@cse.ucdavis.edu} }

\date{\today}

\begin{document}
\maketitle

\begin{abstract}
This paper addresses the matter of inequality in network formation games. We employ a quantity that we are calling the \emph{Nash Inequality Ratio} (NIR), defined as the maximal ratio between the highest and lowest costs incurred to individual agents in a Nash equilibrium strategy, to characterize the extent to which inequality is possible in equilibrium. We give tight upper bounds on the NIR for the network formation games of Fabrikant \etal (PODC '03) and  Ehsani \etal (SPAA '11). With respect to the relationship between equality and social efficiency, we show that, contrary to common expectations, efficiency does not necessarily come at the expense of increased inequality.
\end{abstract}


\section{Introduction}
\label{sec:intro}
Non-cooperative game theory uses the concept of equilibria to capture the idea that, in a competitive world, rational agents will maneuver themselves to a fixed point from which no further maneuvering will yield additional benefits (\eg a lower cost). The most ubiquitous equilibrium concept is the \emph{Nash equilibrium}, which is satisfied when no individual agent can achieve a lower cost by changing their strategy given that the strategies of every other agent remain unchanged. In a Nash equilibrium there can exist a disparity between the costs incurred by individual agents, with the ``more fortunate'' agents subjected to lower costs than the ``less fortunate'' ones. In this manuscript, we investigate the \emph{Nash Inequality Ratio (NIR)}, defined as the maximum value of the ratio between the highest and lowest costs found within a single Nash equilibrium, to determine the extent to which cost disparity can arise between pairs of agents in an equilibrium outcome. 

Recent years have witnessed widespread growing concerns over economic inequality (see, for example, \cite{Bowles2012,Stiglitz2012,Piketty2014}). Yet, within the the algorithmic game theory community -- a field of study sitting at the intersection of computer science, economics, and game theory -- inequality has received very little attention, as evidenced by the nonexistence of a standard vernacular with which to discuss bounds of the sort captured by the NIR. 
The NIR expresses, in a natural way, a fundamental property of a strategic scenario, akin to the well-known metrics used to characterize other qualities of equilibria such as the \emph{Price of Anarchy} (PoA) \cite{KP1999,KP2009} and the \emph{Price of Stability} (PoS) \cite{Anshelevich2004}.

In this manuscript we analyze inequality in network formation games. Network formation games model the formation of network structures by and between a collection of strategic, self-interested individuals or \emph{agents}. In these games, connectivity is deemed to be desirable though costly, and it is up to the agents to reconcile the gains they achieve through additional connectivity (\eg access to information, the ability to communicate and coordinate, etc.) with the costs or resource limitations that limit the number of links that they can afford to create.

We focus on two network formation games: the {\sc Undirected Connections} (UC) game of Fabrikant \etal \cite{Fab2003} and the {\sc Undirected Bounded Budget Connections} (UBBC) game of Ehsani \etal \cite{Ehsani2011}. In both games, there is a set $N = \{1, \dots, n\}$ of strategic agents. Each agent $i \in N$ selects a linking strategy $s_i \subseteq N \setminus \{i\}$ that identifies a subset of other agents which they will build links to. The joint strategy $s = (s_1, \dots, s_n)$ induces an undirected network $G_s = (N, E_s)$ in which the agents themselves are vertices and the edge set $E_s$ is the union of every agent's linking strategies. In both games, agents incur a \emph{usage cost} that is defined for each $i \in N$ to be the sum of  distances between $i$ and every other $j \in N \setminus \{i\}$ in $G_s$. The two games differ in how the cost of building edges is accounted for: in the UC game, each agent $i$ incurs an additional \emph{construction cost}, defined to be $\alpha$ times the number of edges that $i$ builds; and in the UBBC game, each agent $i$ is endowed with an \emph{edge budget}, $k_i > 0$, determining the maximum number of edges that $i$ can build.

Our first results are presented in Section~\ref{sec:upper_bounds} where we establish tight upper bounds on the NIR for these two network formation games. For the UC game, we establish bounds parameterized on the cost of building links, $\alpha$, which is assumed to be a constant. In particular, when links are cheap ($\alpha < 1$), we find that the NIR is at most $1 + \alpha < 2$; and when links are expensive ($\alpha \geq 1$), the NIR is asymptotically bounded (as the number of agents $n \rightarrow \infty$) by $\max \{2, (1+\alpha)/2\}$. Then, for the UBBC game, we prove that the NIR is bounded by the constant $2$. We show that this bound is asymptotically tight for every positive budget allocation -- including those that endow different budgets to different agents. 

With these results established, in Section~\ref{sec:ee} we examine the relationships between efficiency and (in)equality. Here we focus on Nash equilibrium strategies that are also \emph{efficient} (\ie minimize the social cost). In both games we find that this relationship is largely dependent upon the availability of resources (edges). In particular, we find that when resources are scarce (that is, when edges are expensive in the UC model and when budgets are small in the UBBC model), the two games behave quite differently from each other. In the scarce regime for the UBBC game, no efficient Nash equilibrium strategies achieve cost equality, and some actually maximize the inequality ratio (achieving the NIR upper bound established in Section~\ref{sec:upper_bounds}). On the other hand, in the scarce regime for the UC game, there are some efficient Nash equilibrium strategies with egalitarian costs and others with maximal inequality. These results demonstrate that the relationship between equality, efficiency, and equilibrium is entirely model-specific and varies across network formation games.

The research agenda initiated in this manuscript -- using the Nash Inequality Ratio to characterize the relationship between equality, efficiency, and equilibrium -- opens an avenue for future research that involves the bounding of these quantities in different games. Further, we hope that this work stimulates interest in examining questions of inequality and equilibrium more broadly.

Section~\ref{sec:con} concludes with closing remarks and acknowledgements. Finally, to facilitate the flow of the main ideas, tangential proofs are deferred to Appendix~\ref{app:proofs}.

\section{Preliminaries}
\label{sec:prelim}
In this section, we formally define the Nash Inequality Ratio (Section~\ref{sec:prelim:nir}) and the network formation games (Section~\ref{sec:prelim:models}) upon which the remainder of the paper is focused.

\subsection{Nash Inequality Ratio}
\label{sec:prelim:nir}
Studies of economic inequality have spawned a number of metrics that aim to quantify the level of inequality in a given system -- the most well-known being the Gini coefficient \cite{Gini1912} used to quantify the level of inequality in the distribution of utility (wealth) across a population of individuals. The related topic of \emph{fairness} has also received considerable attention in the game theory literature, particularly with respect to mechanism design \cite{Varian1974,Brams1996,AGT2007ch15} where the goal is to develop allocation mechanisms that achieve various notions of equity among the parties involved. Roughgarden \cite{Roughgarden2002} proposed a metric quantifying \emph{unfairness} in the context of non-cooperative routing games as the maximum ratio between an agent's cost (\ie latency of a flow path) in a socially optimal outcome and that of a Nash equilibrium outcome.

In this manuscript, we put forward the \emph{Nash Inequality Ratio (NIR)} as a simple metric that bounds, in a natural way, the extent that inequality between agents is supported in Nash equilibrium outcomes.

Consider a game $\Gamma$ involving a set $N = \{1, \dots, n\}$ of players/agents where, for each agent $i \in N$, $S_i$ specifies $i$'s strategy space. A joint strategy $s = (s_1, \dots, s_n) \in S_1 \times \cdots \times S_n = S$ yields an outcome, and each agent $i$ incurs a cost $c_i(s)$ that is a function of the joint strategy. 

Given a joint strategy $s$, let $\mu = \amin_{i \in N} c_i(s)$ denote an agent that incurs a minimal cost, and $\chi = \amax_{i \in N} c_i(s)$ denote an agent that incurs the highest cost. The \emph{inequality ratio (IR)} for the joint strategy $s$ is defined as:
	\begin{equation*}
	\ir(s) = \frac{c_{\chi}(s)}{c_{\mu}(s)} 
		= \frac{\max_{i \in N} c_i(s)}{\min_{j \in N} c_j(s)}.
	\end{equation*}
That is, the IR of a strategy $s$ is the maximal cost ratio between any pair of agents.

Recall that a joint strategy $s$ is a Nash equilibrium if, for every agent $i \in N$ and every $s_i' \in S_i$, we have
	\begin{equation*}
	c_i(s_i, s_{-i}) \leq c_i(s_i', s_{-i}),
	\end{equation*}
where $s_{-i}$ is shorthand for the joint strategy of every agent except $i$. In a Nash equilibrium, no agent stands to lower their cost through a unilateral strategy update. Let $S_{NE} \subseteq S$ denote the set of all Nash equilibrium. The \emph{Nash Inequality Ratio (NIR)} is defined to be the greatest IR among all Nash equilibrium strategies:
	\begin{equation}
	\label{eq:nir}
	\nir(\Gamma) = \max_{s \in S_{NE}} \ir(s).
	\end{equation}
Note that the NIR is kind of a ``max-max'' quantity, in that the NIR of a game is the maximum inequality ratio among all Nash equilibrium strategies for that game, and each inequality ratio itself is a maximal quantity with respect to a particular strategy.

Finally, by definition we have $\ir(s) \geq 1$ for every strategy $s \in S$. A strategy $s$ is called \emph{egalitarian} if $\ir(s) = 1$.

\subsection{Network Formation Games}
\label{sec:prelim:models}
Network formation games model the creation of networks by rational and self-interested agents strategically building edges between one another. In this section, we describe the two specific network formation games upon which we base the inequality analysis presented in this paper: the {\sc Undirected Connections} game (Section~\ref{sec:prelim:models:uc}) and the {\sc Undirected Bounded Budget Connections} game (Section~\ref{sec:prelim:models:ubbc}). For a comprehensive introduction to network formation games the reader is referred to the surveys by Jackson \cite{Jackson2003,Jackson2008}, Goyal \cite{GoyalBook2007}, and Tardos and Wexler \cite{AGT2007ch19}. 

Questions regarding inequality in network formation games have received very little attention in the literature. To the best of our knowledge, the only paper that looks explicitly at inequality in non-cooperative network formation games is by Goyal and Joshi \cite{Goyal2006}, who analyze the effect of local splillovers in pairwise stable networks. The authors found that agents with more connections and larger neighborhoods earn higher utilities than those with fewer connections and smaller neighborhoods. The results presented in \cite{Goyal2006} show that inequality can arise in equilibrium outcomes, but the authors do not go on to quantify the extent of the inequality.

\subsubsection{The Undirected Connections Game}
\label{sec:prelim:models:uc}
The {\sc Undirected Connections} (UC) network formation game\footnote{The name {\sc Undirected Connections} game is our own choosing in order to reinforce the nature of the game.}  introduced by Fabrikant, Luthra, Maneva, Papadimitriou, and Shenker \cite{Fab2003} was the first network formation game to appear in the algorithmic game theory literature.\footnote{We note that earlier work on network formation games appeared in the economics literature -- \cf \cite{Aumann1988,Jackson1996,Bala2000}.} In addition to introducing the UC model and establishing some basic properties of efficient and Nash equilibrium outcomes, Fabrikant \etal establish bounds on the Price of Anarchy, which have subsequently been tightened (\eg \cite{Albers2014,Demaine2012b,Matus2013,Mamageishvili2013}).

The UC model, as defined by Fabrikant \etal \cite{Fab2003}, is specified by a set $N = \{1, \dots, n\}$ of strategic agents and a parameter $\alpha > 0$, a constant, that determines the cost of building a single edge. The strategy space for an individual agent $i \in N$, denoted $S_i$, consists of all possible subsets of other agents that $i$ can build a direct connection to; \ie $S_i \subseteq \mathcal{P}( N \setminus \{i\} )$.\footnote{We use the notation $\mathcal{P}(X)$ to denote the \emph{power set} of a set $X$; \ie the set of all subsets of $X$.} A (pure) strategy for agent $i$, denoted $s_i \in S_i$, is a specific subset of other agents that $i$ wishes to establish links with. A joint strategy $s = (s_1, s_2, \dots, s_n)$, representing the strategy selections of every agent in $N$, induces an undirected network $G_s = (N, E_s)$, wherein the agents themselves are present as vertices and the edge set is defined to be $E_s = \{ \{i,j\} : j \in s_i \}$. Because an undirected edge $\{i,j\}$ is present in $G_s$ if either $j \in s_i$ or $i \in s_j$, edge formation in the UC model is said to be \emph{unilateral}.

Each agent $i$ incurs a cost that is a function of both its own strategy, $s_i$, and the joint strategy of every other agent, $s_{-i}$. This cost for agent $i$ consists of both a \emph{usage cost}, $d_i(s)$, and a \emph{creation cost}, $b_i(s)$. The usage cost is defined to be the sum of distances between $i$ and every other agent $j$; $d_i(s) = \sum_{j \in N} \ell_{G_s}(i,j)$ where $\ell_{G_s}(i,j)$ denotes the length of the shortest path between nodes $i$ and $j$ in the graph $G_s$ (or $\infty$ if no such path exists).\footnote{We will sometimes use the shorthand $\ell_s(\cdot, \cdot)$ in place of $\ell_{G_s}(\cdot, \cdot)$.} The creation cost is defined to be linear in the number of edges $i$ contributes to the network's construction; specifically, $b_i(s) = \alpha \cdot |s_i|$ where $|s_i|$ conveys the number of edges that $i$ builds in the network and $\alpha \geq 0$ is a constant, specified as a game parameter. Hence, the cost to an agent $i \in N$ given the joint strategy profile $s = (s_i, s_{-i})$ is
	\begin{equation}
	\label{eq:uc:cost}
	c_i(s) = b_i(s) + d_i(s) = \alpha |s_i| + \sum_{j \in N} \ell_{G_s}(i,j).
	\end{equation}

The \emph{social cost} is defined as the sum of the agents' individual costs;
	\begin{equation}
	\label{eq:uc:sc}
	C(s) = \sum_{i \in N} c_i(s).
	\end{equation}
A strategy profile $s \in S$ that minimizes \eqref{eq:uc:sc} is called \emph{efficient}.

Fabrikant \etal identify the topologies of efficient and Nash equilibrium strategies for three different regimes of $\alpha$, which we summarize in Proposition~\ref{prop:fab2003:range}. These results will play prominently in our analysis of the UC game.

\begin{proposition}[\cite{Fab2003}]
\label{prop:fab2003:range}
Efficient and Nash equilibrium outcomes for the UC network formation game are:
	\begin{enumerate}
	\item \label{prop:fab2003:range:1}
	When $\alpha < 1$, the complete graph is both efficient  and the only Nash equilibrium.\footnote{A \emph{complete} graph/network refers to a network in which every pair of nodes are directly linked to each other.}
	
	\item \label{prop:fab2003:range:2}
	When $1 \leq \alpha < 2$, the complete graph is efficient but the star is the only Nash equilibrium.\footnote{\label{foot:star} A \emph{star} refers to a minimally connected network in which there is one central node that is directly linked to the remaining $n-1$ nodes. In a \emph{center-sponsored star}, the cost of all $n-1$ edges is borne by the center node. In a \emph{peripheral-sponsored star}, each of the $n-1$ peripheral agents bares the cost of building the edge connecting them to the center node.}
	
	\item \label{prop:fab2003:range:3}
	When $2 \leq \alpha$, the star is efficient and a Nash equilibrium, although there are other Nash equilibrium outcomes as well.
	\end{enumerate}
\end{proposition}

\subsubsection{The Undirected Bounded Budget Connections Game}
\label{sec:prelim:models:ubbc}
The {\sc Undirected Bounded Budget Connections} (UBBC) network formation game was introduced by Ehsani, Fazli, Mehrabian, Sadeghabad, Safari, Saghafian, and ShokatFadaee \cite{Ehsani2011} as an undirected variant of the {\sc Bounded Budget Connections} (BBC) game of Laoutaris \etal \cite{Laoutaris2008}. Ehsani \etal establish upper bounds on the Price of Anarchy for the UBBC game, building upon the techniques developed by Alon \etal \cite{Alon2010} for a related network formation game.

As with the UC game, the UBBC game involves a set $N = \{1, \dots, n\}$ of strategic agents building edges between one another. However, in the UBBC model, each agent $i \in N$ is endowed with a \emph{budget} $k_i > 0$ that determines the maximum number of edges the agent can create. The budgets $k_1, k_2, \dots, k_n$ are specified exogenously and, in general, need not be identical. We refer to the special case when all edge budgets are identical (\ie $\forall i \in N, k_i = k$) as \emph{uniform} instances of the UBBC game.

Edge formation is unilateral, so a joint strategy $s = (s_1, \dots, s_n)$ induces a network $G_s = (N, E_s)$ with edges $E_s = \{ \{i,j\} : j \in s_i \}$, and the cost to an agent $i$ given $s$ is defined to be
	\begin{equation}
	\label{eq:ubbc:cost}
	c_i(s) = d_i(s) = \sum_{j \in N} \ell_{G_s}(i,j).
	\end{equation}
The social cost for the UBBC game is defined as it is in the UC game; \ie the sum of the costs incurred by the individual agents,
	$$C(s) = \sum_{i \in N} c_i(s).$$

The remainder of this section includes some simple facts regarding Nash equilibria and efficient outcomes of the UBBC game. The first lemma is from Eshani \etal and identifies a sufficient condition for a joint strategy to be a Nash equilibrium.

\begin{lemma}[\cite{Ehsani2011}]
\label{lem:ubbc:diam2}
A UBBC strategy profile $s = (s_1, \dots, s_n)$ that induces a network $G_s$ without parallel edges and a diameter of at most $2$ is a Nash equilibrium.
\end{lemma}

The next two results identify properties of efficient outcomes for uniform instances of the UBBC game. Lemma~\ref{lem:ubbc:sc:eff} discerns the social cost of efficient networks while Proposition~\ref{prop:ubbc:diam2:eff} shows that, for sufficiently sparse instances, efficient outcomes are necessarily networks with a diameter of two. The proofs of Lemma~\ref{lem:ubbc:sc:eff} and Proposition~\ref{prop:ubbc:diam2:eff} can be found in Appendix~\ref{app:proofs}.

\begin{lemma}
\label{lem:ubbc:sc:eff}
The social cost of any efficient strategy profile for a uniform UBBC instance with edge budgets $k \geq 1$ is
	\begin{equation*}
		2 n (n - 1) - 2 n k.
	\end{equation*}
\end{lemma}

\begin{proposition}
\label{prop:ubbc:diam2:eff}
Every uniform UBBC instance with $k < (n-1)/2$ has an efficient outcome with a diameter of $2$.
\end{proposition}

\section{NIR Upper Bounds}
\label{sec:upper_bounds}
This section presents our results on upper bounding the Nash Inequality Ratio for the two network formation games. NIR upper bounds are established for the UC game in Section~\ref{sec:upper_bounds:uc} (see Theorem~\ref{thm:uc:ub}) and for the UBBC game in Section~\ref{sec:upper_bounds:ubbc} (see Theorem~\ref{thm:ubbc:ub}).

\subsection{The Undirected Connections Game}
\label{sec:upper_bounds:uc}
We establish upper bounds on the NIR for the UC game for two regimes of $\alpha$: when $\alpha < 1$ and $\alpha \geq 1$. These bounds are stated in Theorem~\ref{thm:uc:ub}, showing that inequality is independent of the number of agents when $\alpha$ is a constant. Our upper bound implies that when $\alpha > 3$, inequality scales linearly with $\alpha$.

\begin{theorem}
\label{thm:uc:ub}
Upper bounds on the NIR for the UC game:
	\begin{enumerate}
	\item \label{thm:uc:ub:1}
	When $\alpha < 1$, the NIR is at most $1 + \alpha < 2$.
	
	\item \label{thm:uc:ub:4}
	When $1 \leq \alpha < \infty$ is a constant (independent of $n$), the NIR is at most $\max \{2, \frac{1+\alpha}{2}\}$ in the limit as $n \rightarrow \infty$.
	\end{enumerate}
\end{theorem}

We will prove the two parts of Theorem~\ref{thm:uc:ub} separately in Lemma~\ref{lem:uc:ub:1} (next) and Lemma~\ref{lem:uc:ub:4} (below).

\begin{lemma}
\label{lem:uc:ub:1}
When $\alpha < 1$, the NIR for the UC games is at most $1 + \alpha < 2$.
\end{lemma}

\begin{proof}
From Proposition~\ref{prop:fab2003:range} we know that whenever $\alpha < 1$ then the only Nash equilibrium is a complete graph, and every agent is adjacent to $n - 1$ other agents. In the complete graph, the usage costs incurred by every agent is the same, and the only disparity that can arise is due to the agents' construction costs. In the most extreme case, the min-cost agent does not buy any links and the max-cost agent buys $m > 0$ links. The inequality ratio is therefore
	\begin{equation*}
	\frac{ (n - 1) + \alpha m}{n - 1} = 1 + \frac{\alpha m}{n - 1} \leq 1 + \alpha,
	\end{equation*}
where $m$ can be at most $n-1$ (\ie $m = n-1$ when the max-cost agent buys all of their incident links). It is easy to see that the max-cost agent is not inclined to discard any of their edges since doing so would increase their usage cost by $1$ but only yield a savings of $\alpha < 1$ in their construction cost, resulting in a net cost increase.
\end{proof} 

Before addressing the second case of Theorem~\ref{thm:uc:ub}, we first need a couple of lemmas. In Lemma~\ref{lem:fab2003:tree}, we give an expression that captures the inequality ratio for an arbitrary Nash equilibrium of the UC game, and then show in Corollary~\ref{cor:uc:star-max} that this quantity is maximized in star networks. Then, in Lemma~\ref{lem:fab2003:star_inequality}, we bound the inequality ratio of star network topologies.

Lemma~\ref{lem:fab2003:tree} and Corollary~\ref{cor:uc:star-max} require additional notation. We refer to a min-cost and max-cost agent as $\mu \in \amin_{i \in N} c_i(s)$ and $\chi \in \amax_{i \in N} c_i(s)$, respectively. Note that, in general, there may exist multiple min- and max-cost agents in any given Nash equilibrium, but in the following, we assume that $\mu$ and $\chi$ refer to the same min- and max-cost agents throughout the argument.
For an agent $i$ and a joint strategy $s$, let $T_s(i)$ denote a \emph{shortest-path tree} rooted at $i$ that is built from a breadth-first traversal of $G_s$ beginning at node/agent $i$.\footnote{Arguments using such trees are commonly used in the analysis of the Price of Anarchy for network formation games; see, for example, Albers \etal \cite{Albers2014}.} Edges of $G_s$ that appear in $T_s(i)$ are referred to as \emph{tree edges}, and those that do not are \emph{non-tree edges}. By construction, for any agent $j$ that appears in layer $k \geq 0$ of $T_s(i)$, it follows that $\ell_s(i,j) = k$. Also, for any agent $j$ appearing in layer $k$ and every non-tree edge $\{j, j'\}$, it follows (from the construction of $T_s(i)$) that agent $j'$ must be in layer $k' \in \{k-1, k, k+1\}$. 

For a pair of agents $i, j$, let $T_s(j ; i)$ denote the subtree of $T_s(i)$ that is rooted at $j$ such that, without loss of generality, all agents $j'$ that are adjacent to $j$ in $G_s$ and are in layer $\ell_s(i,j) + 1$ of $T_s(i)$, are present in the first layer of the subtree $T_s(j ; i)$.
With a slight abuse of notation, we will use $T_s(j ; i)$ to refer to the set of agents in the subtree as well as the subtree itself. 
Finally, let $\bar{T}_s(j ; i)$ refer to the subtree (and its constituent agents) of $T_s(i)$ that are not a part of $T_i(j ; i)$. See Figure~\ref{fig:fab-tree} for an illustration with $i = \mu$ and $j = \chi$ and solid (dotted) lines representing tree (non-tree) edges.

\begin{lemma}
\label{lem:fab2003:tree}
Let $\mu$ and $\chi$ be min- and max-cost agents in a Nash equilibrium $s$. The inequality ratio is at most
	\begin{equation}
	\label{eq:fab2003:tree}
	\ir(s) = \frac{c_\chi(s)}{c_\mu(s)} \leq \frac{\alpha \left( 1 + | s_\chi \cap T_s({\chi;\mu}) | \right) + 
			| \bar{T}_s(\chi;\mu) | 
			- \ell_s(\chi, \mu) \cdot | T_s(\chi;\mu) | + d_\mu(s)}{d_\mu(s)},
	\end{equation}
where $d_\mu(s) = \sum_i \ell_s(\mu,i)$ is the distance-cost for agent $\mu$.
\end{lemma}

\begin{figure}[t]
	\centering
	\includegraphics{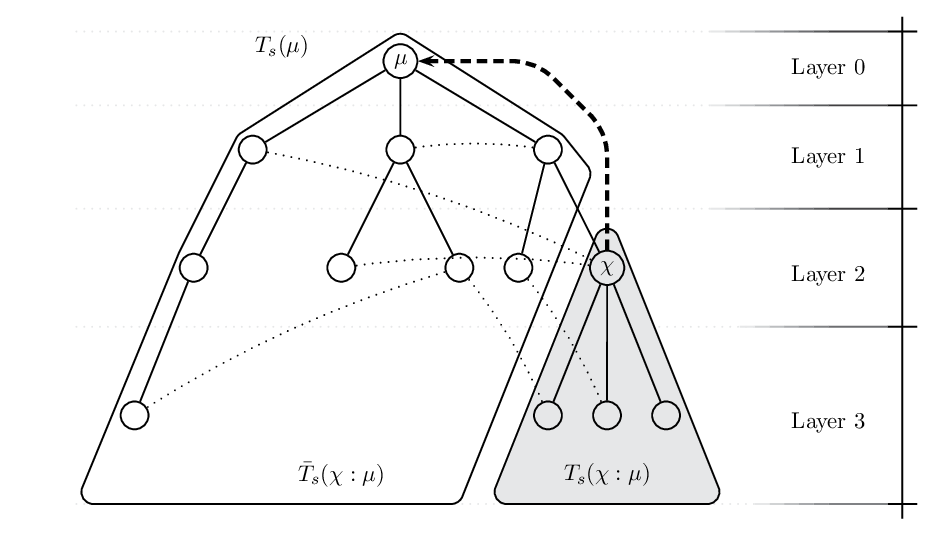}
	\caption{An induced shortest-path tree, $T_s(\mu)$, of a graph $G_s$ rooted at a min-cost agent $\mu$. Solid edges represent tree-edges, and dotted edges represent non-tree-edges. Given a max-cost agent $\chi \neq \mu$ in layer $k = \ell_s(\mu, \chi)$, the tree $T_s(\chi; \mu)$ is the subtree rooted at $\chi$ with the stipulation that every $i$ in layer $k+1$ that is adjacent to $\chi$ in $G_s$ is in $T_s(\chi; \mu)$. $\bar{T}_s(\chi;\mu)$ represents the part of $T_s(\mu)$ made up of agents/nodes that are not a part of $T_s(\chi;\mu)$.}
	\label{fig:fab-tree}
\end{figure}

Intuitively, the bound in Equation~\ref{eq:fab2003:tree} follows from deriving upper- and lower-bounds on the costs of the max- and min-cost agents, respectively. The upper bound on the cost of the max-cost agent, $\chi$, follows from considering a strategic deviation whereby $\chi$ builds a direct link to a min-cost agent $\mu$ and then drops all of its purchased links except those to agents $j \in s_\chi$ that are further from $\mu$ than $\chi$ is. If the original strategy is a Nash equilibrium, then we know that $\chi$'s cost must be no higher than its cost given this new strategy.\footnote{Our derivation of this upper bound on $\chi$'s cost extends a similar argument used by Albers \etal \cite{Albers2014}. In particular, compare our Equation~\eqref{eq:fab2003:tree:1} below with Equation~(2) in \cite{Albers2014}.} The cost lower bound for a min-cost agent, $\mu$, follows from simply assuming that $\mu$ did not purchase any edges (\ie $c_\mu(s) \geq d_\mu(s)$).

\begin{proof} 
Let $\chi$ and $\mu$ be a max- and min-cost agent, respectively, in a Nash equilibrium strategy $s$, and let $T_s(\mu)$, $T_s(\chi;\mu)$, and $\bar{T}_s(\chi;\mu)$ be as defined above. Consider the strategy $s_{\chi}'$ for agent $\chi$ that is obtained from $s_\chi$ by dropping all of $\chi$'s edges to agents outside of $T_s(\chi;\mu)$ and buying an edge to $\mu$ (see Figure~\ref{fig:fab-tree}). Agent $\chi$'s cost under this new strategy $s' = (s_{\chi}', s_{-\chi})$ is
	\begin{equation}
	\label{eq:fab2003:tree:1}
	c_{\chi}(s') \leq \underbrace{\alpha \left( 1 + | s_\chi \cap T_s(\chi;\mu) | \right)}_{b_\chi(s')}
				+ \underbrace{| \bar{T}_s(\chi;\mu) | 
					- \ell_s(\mu, \chi) \cdot | T_s(\chi;\mu) | + d_\mu(s)}_{\text{upper-bound on } d_\chi(s')}.
	\end{equation}
Because $s$ is a Nash equilibrium, we know that $c_\chi(s) \leq c_\chi(s')$, and the lemma follows from \eqref{eq:fab2003:tree:1} and the fact that $c_\mu(s) \geq d_\mu(s)$.

Agent $\chi$'s usage cost, $b_\chi(s')$, simply reflects the construction of $s_\chi'$ -- building an edge to $\mu$ and dropping all edges to nodes in $s_\chi$ that are outside of $T_s(\chi; \mu)$. The remainder of the proof is dedicated to establishing the upper bound on $d_\chi(s')$. For a set $X \subset N$, let $d_i^X(s) = \sum_{j \in X} \ell_s(i, j)$ denote agent $i$'s usage cost to agents in $X$ given the joint strategy $s$. We will bound $d_\chi(s')$ in terms of $d_\mu(s) = d_\mu^{T_s(\chi; \mu)}(s) + d_\mu^{\bar{T}_s(\chi; \mu)}(s)$. Notice that 
	\begin{eqnarray}
	d_\mu^{T_s(\chi;\mu)}(s) 
	&=& \sum_{i \in T_s(\chi;\mu)} \left( \ell_s(\mu, \chi) + \ell_s(\chi, i) \right) \nonumber \\
	&=& \ell_s(\mu, \chi) \cdot |T_s(\chi;\mu)| + \sum_{i \in T_s(\chi;\mu)} \ell_s(\chi, i) \nonumber \\
	&=& \ell_s(\mu, \chi) \cdot |T_s(\chi;\mu)| + d_\chi^{T_s(\chi;\mu)}(s) \nonumber
	\end{eqnarray}
since $\chi$ is on a shortest-path between $\mu$ and every agent $i \in T_s(\chi; \mu)$. Rearranging, we get
	\begin{equation*}
	d_\chi^{T_s(\chi;\mu)}(s) = d_\mu^{T_s(\chi;\mu)}(s) - \ell_s(\mu, \chi) \cdot |T_s(\chi;\mu)|.
	\end{equation*}
In $s'$, $\chi$'s distance to agents in $T_s(\chi;\mu)$ is unchanged from $s$, so we can express $\chi$'s usage cost to agents in $T_s(\chi;\mu)$ under $s'$ in relation to $\mu$'s distance to agents in $T_s(\chi;\mu)$ under $s$ as
	\begin{equation}
	\label{eq:fab2003:tree:2}
	d_\chi^{T_s(\chi;\mu)}(s') = d_\mu^{T_s(\chi;\mu)}(s) - \ell_s(\mu, \chi) \cdot |T_s(\chi;\mu)|.
	\end{equation}
	
Next, we can derive $d_\chi^{\bar{T}_s(\chi;\mu)}(s')$ in terms of $d_\mu^{\bar{T}_s(\chi;\mu)}(s)$ in a very straight-forward way. By including a link to $\mu$ in $s_\chi'$ and dropping all links to agents outside of $T_s(\chi; \mu)$, $\chi$'s usage cost to agents in $\bar{T}_s(\chi;\mu)$ is at most $|\bar{T}_s(\chi;\mu)|$ more than $\mu$'s usage cost to agents in $\bar{T}_s(\chi;\mu)$ under $s$; \ie 
	\begin{equation}
	\label{eq:fab2003:tree:3}
	d_\chi^{\bar{T}_s(\chi;\mu)}(s') \leq |\bar{T}_s(\chi;\mu)| + d_\mu^{\bar{T}_s(\chi;\mu)}(s).
	\end{equation}
Putting \eqref{eq:fab2003:tree:2} and \eqref{eq:fab2003:tree:3} together, we get
	\begin{eqnarray*}
	d_\chi(s') &=& d_\chi^{T_s(\chi;\mu)}(s') + d_\chi^{\bar{T}_s(\chi;\mu)}(s') \\
			&\leq& d_\mu^{T_s(\chi;\mu)}(s) - \ell_s(\mu, \chi) \cdot |T_s(\chi;\mu)| 
					+ |\bar{T}_s(\chi;\mu)| + d_\mu^{\bar{T}_s(\chi;\mu)}(s) \\
			&=& d_\mu(s) - \ell_s(\mu, \chi) \cdot |T_s(\chi;\mu)| + |\bar{T}_s(\chi;\mu)|,
	\end{eqnarray*}
where we substituted $d_\mu(s) = d_\mu^{T_s(\chi;\mu)}(s) + d_\mu^{\bar{T}_s(\chi;\mu)}(s)$ in the last line. This completes the proof.
\end{proof} 

\begin{corollary} 
\label{cor:uc:star-max}
Among Nash equilibrium strategies that arise when the edge construction cost is a constant $1 \leq \alpha < \infty$, the star topology maximizes the inequality ratio.
\end{corollary}

\begin{proof} 
We will show that the inequality ratio established in Lemma~\ref{lem:fab2003:tree} is maximized in a star network topology.
Recall Equation~\eqref{eq:fab2003:tree}, bounding the inequality ratio of an arbitrary Nash equilibrium strategy $s$ for the UC game, repeated here for convenience: 
	\begin{equation} 
	\ir(s) \leq \frac{\alpha \left( 1 + | s_\chi \cap T_s({\chi;\mu}) | \right) + 
			| \bar{T}_s(\chi;\mu) |
			- \ell_s(\chi, \mu) \cdot | T_s(\chi;\mu) | 
			+ d_\mu(s)}
		{d_\mu(s)}. \tag{\ref{eq:fab2003:tree}}
	\end{equation}
Let $x = |T_s(\chi;\mu)|$ (ergo, $|\bar{T}_s(\chi;\mu)| = n - x$) and $x_0 = 1 + | s_\chi \cap T_s(\chi;\mu) | \leq x$. Substituting 
	\begin{eqnarray*}
		d_\mu(s) &=& d_\mu^{T_s(\chi;\mu)}(s) + d_\mu^{\bar{T}_s(\chi;\mu)}(s) \\
		&=& \ell_s(\chi, \mu) x + d_\chi^{T_s(\chi;\mu)}(s) + d_\mu^{\bar{T}_s(\chi;\mu)}(s),
	\end{eqnarray*}
we can rewrite \eqref{eq:fab2003:tree} as
	\begin{eqnarray}
	\ir(s) &\leq& \frac{ \alpha x_0 + (n - x) - \ell_s(\chi, \mu) x + d_\mu(s)}{d_\mu(s)} \nonumber \\
	&=& \frac{\alpha x_0 + n - x - \ell_s(\chi, \mu) x + d_\mu^{T_s(\chi;\mu)}(s) + d_\mu^{\bar{T}_s(\chi;\mu)}(s)}{d_\mu^{T_s(\chi;\mu)}(s) + d_\mu^{\bar{T}_s(\chi;\mu)}(s)} \nonumber \\
	&=& \frac{\alpha x_0 + n - x - \ell_s(\chi, \mu) x + \ell_s(\chi, \mu) x + d_\chi^{T_s(\chi;\mu)}(s) + d_\mu^{\bar{T}_s(\chi;\mu)}(s)}{\ell_s(\chi, \mu) x + d_\chi^{T_s(\chi;\mu)}(s) + d_\mu^{\bar{T}_s(\chi;\mu)}(s)} \nonumber \\
	\label{eq:uc:star-max:2}
	&=& \frac{\alpha x_0 - n - x + d_\chi^{T_s(\chi;\mu)}(s) + d_\mu^{\bar{T}_s(\chi;\mu)}(s)}{\ell_s(\chi, \mu) x + d_\chi^{T_s(\chi;\mu)}(s) + d_\mu^{\bar{T}_s(\chi;\mu)}(s)}.
	\end{eqnarray}
Notice that in Equation~\eqref{eq:uc:star-max:2}, the distance between $\mu$ and $\chi$ in $G_s$ only appears in the denominator. Therefore, toward our aim of maximizing \eqref{eq:uc:star-max:2}, we can infer that $\ell_s(\chi, \mu) = 1$. Hence, $\mu$ and $\chi$ are necessarily adjacent in a Nash equilibrium $s$ that maximizes the inequality ratio. With $\mu$ and $\chi$ adjacent, we can substitute
	\begin{equation*}
	d_\mu^{\bar{T}_s(\chi;\mu)}(s) = d_\chi^{\bar{T}_s(\chi;\mu)}(s) - | \bar{T}_s(\chi;\mu) | 
	= d_\chi^{\bar{T}_s(\chi;\mu)}(s) - n + x
	\end{equation*}
and
	\begin{equation*}
	d_\chi^{T_s(\chi;\mu)}(s) = d_\mu^{T_s(\chi;\mu)}(s) - | T_s(\chi;\mu) |
	= d_\mu^{T_s(\chi;\mu)}(s) - x,
	\end{equation*}
giving us
	\begin{equation}
	\label{eq:uc:star-max:3}
	\ir(s) \leq \frac{\alpha x_0 + d_\chi(s) - 2x - n}{d_\mu(s) + 2x - n}.
	\end{equation}	
	
This leaves us with two approaches to identify topologies that maximize the inequality ratio: we can maximize $x_0$ (\ie by setting $x_0 = n - 1$), or we can minimize $x$ (\ie set $x = 1$). Both approaches imply star network topologies (the center- and peripheral-sponsored stars, respectively), completing the proof.
\end{proof} 

The next lemma establishes asymptotically tight upper bound on the inequality ratio for star network topologies when the edge cost $1 \leq \alpha < \infty$.

\begin{lemma}
\label{lem:fab2003:star_inequality}
As the number of agents grows toward infinity, the maximal inequality in a star topology for the UC network formation game with a constant edge cost $1 \leq \alpha < \infty$ is $\max \{2, (1 + \alpha) / 2\}$.
\end{lemma}

The proof of Lemma~\ref{lem:fab2003:star_inequality} is straight-forward, so we provide only a sketch here. The full proof can be found in Appendix~\ref{app:proofs}.

\begin{proof}[Proof Sketch]
We consider two star topologies: the center-sponsored star strategy and the peripheral-sponsored star strategy. It is easy to see that the inequality in these two star strategies dominate that of all other star strategies (see the full proof in Appendix~\ref{app:proofs} for details).

In the center-sponsored star strategy, $s^{c}$, the max-cost agent $\chi$ is in the center and the min-cost agent $\mu$ is in the periphery. Therefore, the inequality ratio is
	\begin{equation}
	\label{eq:fab2003:star_inequality:s1}
	\ir(s^{c}) = \frac{c_\chi(s^{c})}{c_\mu(s^{c})}
		= \frac{(\alpha + 1) (n - 1)}{2n - 3}. \nonumber
	\end{equation}
In the peripheral-sponsored star strategy, $s^{p}$, the min-cost agent $\mu$ is in the center and the max-cost agent $\chi$ is in the periphery. Hence, the inequality ratio is
	\begin{equation}
	\label{eq:fab2003:star_inequality:s2}
	\ir(s^{p}) = \frac{c_\chi(s^{p})}{c_\mu(s^{p})}
		= \frac{\alpha + 2n - 3}{n - 1}. \nonumber
	\end{equation}
As $n \rightarrow \infty$, $\ir(s^{c})$ approaches $(1 + \alpha) / 2$ and $\ir(s^{p})$ approaches $2$. Therefore, the inequality ratio for a star network is $\max \{ 2, (1 + \alpha) / 2\}$ as $n$ grows to infinity.
\end{proof} 

Finally, we are ready to prove the second case of Theorem~\ref{thm:uc:ub}, which will follow directly from Lemmas~\ref{lem:fab2003:tree} and \ref{lem:fab2003:star_inequality} and Corollary~\ref{cor:uc:star-max}.

\begin{lemma} 
\label{lem:uc:ub:4}
When $1 \leq \alpha < \infty$ is a constant (independent of $n$), the NIR for the UC game is at most $\max \{2, (1 + \alpha) / 2\}$ in the limit as $n \rightarrow \infty$.
\end{lemma}

\begin{proof}
From Corollary~\ref{cor:uc:star-max} we know that, among Nash equilibrium strategies, the star maximizes the inequality ratio; and Lemma~\ref{lem:fab2003:star_inequality} provides the desired IR upper bound for stars.
\end{proof}

\subsection{The Undirected Bounded Budget Connections Game}
\label{sec:upper_bounds:ubbc}
For the UBBC game, the NIR is at most two (Theorem~\ref{thm:ubbc:ub}). This bound holds for uniform instances of the game, in which every agent has the same edge budget, as well as general cases in which agents can have heterogenous budgets. This upper bound is tight, demonstrating that inequality at equilibrium in the UBBC game is independent of \emph{ex ante} inequalities in edge endowments.

\begin{theorem}
\label{thm:ubbc:ub}
The NIR for the UBBC game is at most $2$.
\end{theorem}
The proof of Theorem~\ref{thm:ubbc:ub} (see below) proceeds as follows: we argue that in any strategy profile where the agents with the maximum and minimum costs are adjacent, the inequality ratio must be strictly less than two. If, on the other hand, the min-cost and max-cost agents are \emph{not} adjacent in a Nash equilibrium, it must be the case that neither of them stand to reduce their individual costs by switching to a strategy that includes a link with the other agent, since doing so would mean that the strategy is not a Nash equilibrium in the first place. This implies that either:
	\begin{enumerate}
	\item \label{thm:ubbc:ub:outline1} The inequality ratio is already less than  $2$; or
	\item \label{thm:ubbc:ub:outline2} A maximum cost agent switching to a strategy that includes a link to a minimum cost agent would lead them to a \emph{higher} cost (\ie by disconnecting the network). 
	\end{enumerate}
However, by Lemma~\ref{lem:ubbc:cut}, case \ref{thm:ubbc:ub:outline2} cannot be a Nash equilibrium, so the NIR can be at most $2$.

\begin{lemma}
\label{lem:ubbc:cut}
Suppose $s = (s_1, \dots, s_n)$ is a UBBC strategy in which there exist (distinct) agents $x,y,z \in N$ with $x \in s_z$, $y \notin s_z$, and $z \notin s_y$ such that:
	\begin{enumerate}
	\item Every $x \rightarrow y$ path in $G_s$ contains agent $z$ as an intermediate node, and
	\item If $z$ were to swap out $x$ for a connection with agent $y$, then the network will become disconnected (\ie the edge $\{z,x\}$ is a bridge).
	\end{enumerate}
Then $s$ cannot be a Nash equilibrium.
\end{lemma}

The proof of Lemma~\ref{lem:ubbc:cut} employs the main result of \cite{Mihalak2012}, regarding the structure of networks in asymmetric swap-equilibrium. The asymmetric swap-equilibrium holds for a strategy $s$ when, for every agent $i \in N$ and every deviation $s_i'$ that differs from $s_i$ with the addition, removal, or swap of a single edge, we have $c_i(s_i, s_{-i}) \leq c_i(s_i', s_{-i})$. The asymmetric swap-equilibrium is a weaker equilibrium concept than the Nash equilibrium for the UBBC, so their result (restated in Theorem~\ref{thm:Mihalak2012} below) carries over to Nash equilibrium for the UBBC. Recall that a graph is $k$-edge-connected if the removal of any $k-1$ edges does not disconnect it. A \emph{$k$-edge-connected component} of a graph $G$ is a maximal subgraph $G' \subset G$ that is $k$-edge-connected.

\begin{theorem}[\cite{Mihalak2012}]
\label{thm:Mihalak2012}
Every network in an asymmetric swap-equilibrium has at most one 2-edge-connected component.
\end{theorem}

Put another way, Theorem~\ref{thm:Mihalak2012} states that a Nash equilibrium strategy $s$ will not induce a graph $G_s$ in which the removal of a single edge will split the graph into two components that are both $2$-edge-connected; \ie the two components cannot both contain cycles.

\begin{figure}[t]
	\centering
	\includegraphics{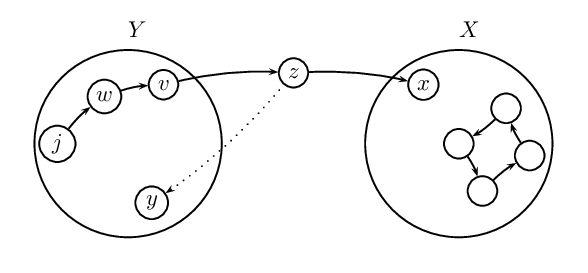}
	\caption{A schematic of the conditions expressed in Lemma~\ref{lem:ubbc:cut}.}
	\label{fig:ubbc:cutlemma}
\end{figure}

\begin{proof}[Proof of Lemma~\ref{lem:ubbc:cut}]
Let $s$ be a strategy with agents $x,y,z \in N$ that satisfy the given conditions and assume (toward a contradiction) that $s$ is a Nash equilibrium. Denote by $X$ the component of the network containing $x$ that would result from the removal of $x$ from agent $z$'s strategy, and let $Y = N \setminus X \cup \{z\}$ be the set of the remaining agents (not in $X$) also excluding agent $z$ (see Figure~\ref{fig:ubbc:cutlemma}). Recall that every agent $i$ contributes at least one edge to $G_s$ (since $k_i \geq 1$) so the component $X$ must contain a cycle. Hence, by Theorem~\ref{thm:Mihalak2012}, we know that the subgraph induced by agents in $Y$ must not contain a cycle.

Let $G_Y$ denote the subgraph induced by the agents in $Y$. If $\exists j \in Y$ such that $k_j > 1$ then either $|s_j| < k_j$ and $j$ is not playing a best response (since they could build an edge to an non-neighboring agent and decrease their cost) or $G_Y$ will contain a cycle; both scenarios contradicting our assumption that $s$ is a Nash equilibrium. Therefore, assume that the subgraph $G_Y$ is a tree, and that $k_j = 1$ for all $j \in Y$. Clearly, for every path $p = (y_1, y_2, \dots, y_m = z)$ that begins at an agent $y_1 \in Y$ and ends at $z$, we have $y_{i+1} \in s_{y_i}$ (to be otherwise would imply that $k_j > 1$ for some $j \in Y$).

Let $j \in Y$ be a leaf in $G_Y$ that is furthest from $z$. Let $w \in s_j$ be $j$'s parent, and let $v \in s_w$ be $j$'s grandparent. If $w = z$ then we are done because this contradicts our assumption that $y$ and $z$ are not adjacent in $G_s$. So assume that $w \neq z$. Let $D_w$ be $w$'s descendants in $G_Y$ (\ie the set of agents for whom $w$ lies on their unique path to $z$), and let $C_w \subseteq D_w$ be $w$'s children (\ie the set of agents in $D_w$ that are adjacent to $w$). Consider a deviation by $j$ to the strategy $s_j' = \{v\}$ that replaces $j$'s link to their parent with a link to their grandparent. This deviation will increase $j$'s distance to $w$ and all nodes in $C_w \setminus \{j\}$ by one and decrease $j$'s distance to every node in $N \setminus C_w \setminus \{j\}$ by one. If $|C_w| < |N \setminus C_w| - 1$ then agent $j$ will benefit from switching to strategy $s_j'$, contradicting our assumption that $s$ is a Nash equilibrium. 

Suppose that $|C_w| \geq |N \setminus C_w| - 1$, which implies that $|C_w| \geq (n-1)/2$. Let $u \in s_v$ be $w$'s grandparent in $G_Y$. If there is another agent $w'$ with $v \in s_{w'}$ then this agent can swap their link to $v$ with a link to $w$ and receive a strictly lower cost, which implies that $s$ was not a Nash equilibrium. If there is no such $w'$, then $w$ can swap their link to $v$ with a link to their grandparent $u$ and receive a strictly lower cost. This again implies that $s$ was not a Nash equilibrium. This completes the proof.
\end{proof} 

We are now ready to attend to the proof of Theorem~\ref{thm:ubbc:ub}. We will use the shorthand $G - \{i,j\}$ to denote the graph $G$ with the edge $\{i,j\}$ removed.

\begin{proof}[Proof of Theorem~\ref{thm:ubbc:ub}]
Consider (toward a contradiction) a strategy profile $s$ that is a Nash equilibrium in which $c_{max}(s) \geq 2 \cdot c_{min}(s)$; \ie a strategy $s$ such that the inequality ratio is at least $2$. Notice that in this strategy an agent $\chi$ with $c_\chi(s) = c_{max}(s)$ cannot be directly connected to an agent $\mu$ with $c_\mu(s) = c_{min}(s)$, because if it were then $\chi$ would be connected to the $n-2$ other agents via $\mu$'s shortest paths for a cost that is at most $n-2$ more than the cost $\mu$ is subjected to. That is, if $\chi \in s_\mu$ or $\mu \in s_\chi$ then 
	\begin{equation}
	\label{eq:ubbc:ub:1}
	c_\chi(s) \leq (n - 2) + c_\mu(s) < 2 \cdot c_\mu(s).
	\end{equation}
The strict inequality in \eqref{eq:ubbc:ub:1} comes from the fact that $c_\mu(s) > n-2$ since it is impossible for $\mu$ to connect to $n-1$ other agents for a cost any less than $n-1$. Therefore, in order for $c_\chi(s) \geq 2 \cdot c_\mu(s)$, it must be the case that $\mu \notin s_\chi$ and $\chi \notin s_\mu$. This implies that $c_\chi(s)$ must be no greater than $(n-2) + c_{\mu}(s)$, since switching to a strategy $s_{\chi}'$ that includes a link to agent $\mu$ would ensure as much, provided that switching to $s_{\chi}'$ does not disconnect the network.

Suppose that $c_\chi(s) > (n-2) + c_\mu(s)$ and if $\chi$ were to switch to a strategy $s_{\chi}'$ that is obtained by swapping out some $x \in s_\chi$ for a link to $\mu$ results in the network $G_{(s_{\chi}', s_{-\chi})}$ becoming disconnected. Then it must be that every $x \rightarrow \mu$ path in $G_s$ includes agent $\chi$, in which case it follows that the component of $G_{s} - \{\chi, x\}$ including $x$ contains a cycle since every agent's budget affords them at least one edge. However, by Lemma~\ref{lem:ubbc:cut}, such a strategy $s$ cannot be a Nash equilibrium, contradicting our assumption that $s$ is a Nash equilibrium. This completes the proof.
\end{proof} 

The upper bound established by Theorem~\ref{thm:ubbc:ub} is asymptotically tight. This can be observed in a ``star-like'' network in which every agent has a link to the min-cost agent $\mu$, and the degree of the max-cost agent $\chi$ is bounded by a constant independent of the number of agents, $n$.\footnote{For a concrete example, consider a uniform instance of the UBBC game in which $k_i = 1$ for all $i \in N$, and the Nash equilibrium strategy that forms a star network with one ``extra'' edge linking a pair of peripheral agents.} In the asymptotic limit, as the number of agents tends toward infinity, the inequality ratio between $\chi$ and $\mu$ equals $2$.

\section{Equality, Equilibrium, and Efficiency}
\label{sec:ee}
With the upper bounds on inequality in Nash equilibrium strategies established, we are now in a position to address the relationship between (in)equality, equilibrium, and efficiency. Although our focus here is on a specific metric of inequality (the inequality ratio) and how it relates to efficiency in a particular setting (the network formation games), the connection between inequality and efficiency more generally is a matter of considerable interest among economic analysts, researchers, and policy makers (\cf \cite{Baland1997,Baland1998,Duclos2006,Bourguignon2007,Bowles2012}), as well as the general public \cite{Stiglitz2012,Piketty2014}. As in Section~\ref{sec:upper_bounds}, the presentation of this section is divided into two parts, with the UC game analyzed in \ref{sec:ee:uc} (see Theorem~\ref{thm:uc:ee}) followed by the UBBC in Section~\ref{sec:ee:ubbc} (see Theorem~\ref{thm:ubbc:ee}).

\subsection{The Undirected Connections Game}
\label{sec:ee:uc}
Theorem~\ref{thm:uc:ee} summarizes our results on the relationship between equality, efficiency, and Nash equilibrium for the UC game for the three regimes of $\alpha$ identified in Proposition~\ref{prop:fab2003:range}. 

\begin{theorem}
\label{thm:uc:ee}
The relationship between inequality, Nash equilibrium, and efficiency in the UC game:
	\begin{enumerate}
	\item \label{thm:uc:ee:1}
	When $\alpha < 1$, there exist efficient Nash equilibrium strategies that maximize the inequality ratio, and other efficient Nash equilibrium strategies that achieve cost equality in the limit as $n \rightarrow \infty$.
	
	\item \label{thm:uc:ee:2}
	When $1 \leq \alpha < 2$, no Nash equilibrium strategy is also efficient. However, among Nash equilibrium strategies, there exist some that maximize the inequality ratio, while others are egalitarian in the limit as $n \rightarrow \infty$.
	
	\item \label{thm:uc:ee:3}
	When $2 \leq \alpha < \infty$, there exist Nash equilibrium strategies that are 
		\begin{itemize}
		\item both efficient and egalitarian in the limit as $n \rightarrow \infty$; and
		\item efficient with maximal inequality in the limit as $n \rightarrow \infty$.
		\end{itemize}
	\end{enumerate}
\end{theorem}

We will prove the three parts of Theorem~\ref{thm:uc:ee} individually in the next three lemmas. 

\begin{lemma}
\label{lem:uc:ee:1}
When $\alpha < 1$, there exist efficient Nash equilibrium strategies for the UC game that maximize the inequality ratio, and other efficient Nash equilibrium strategies that achieve cost equality in the limit as $n \rightarrow \infty$.
\end{lemma}

\begin{proof}
By Proposition~\ref{prop:fab2003:range}, when $\alpha < 1$, both Nash equilibrium strategies and socially efficient strategies result in the formation of the complete network. From this fact and Theorem~\ref{thm:uc:ub}, it follows that there exist socially efficient strategies which maximize the inequality ratio. Therefore, we only need to show that there exist equilibrium strategies that achieve cost equality.

Equality among all agents would follow if every agent builds exactly $(n-1) / 2$ edges, which can only occur when $n$ is odd. When, on the other hand, $n$ is even, then $(n-1)/2$ is non-integral; so the closest we can get to equality calls for a strategy $s$ in which the max-cost agents each build $\lceil (n-1)/2 \rceil$ edges while the min-cost agents get away with building one fewer edge each. In this case, a min-cost agent's cost is
	$c_{min}(s) = c_{max}(s) - \alpha,$
and the inequality ratio is
	\begin{eqnarray}
	\ir(s) &=& \frac{c_{max}(s)}{c_{max}(s) - \alpha} \nonumber \\
		&=& \frac{ (n-1) - \lceil (n-1)/2 \rceil \alpha }{ (n-1) - (\lceil (n-1)/2 \rceil - 1) \alpha } \nonumber \\
		&=& 1 - \frac{ \alpha }{ \alpha + \alpha \lfloor (1-n)/2 \rfloor + n - 1}. \label{eq:uc:ee:1:1}
	\end{eqnarray}
As $n \rightarrow \infty$, Equation~\eqref{eq:uc:ee:1:1} approaches $1$. Hence, inequality vanishes as $n$ increases, concluding the proof.
\end{proof} 

Turning to the second item of Theorem~\ref{thm:uc:ee}, which addresses the case that $\alpha \in [1,2)$, we first note that Proposition~\ref{prop:fab2003:range} already implies that no Nash equilibrium strategy is also an efficient strategy. This is because the Nash equilibrium strategies for this regime of $\alpha$ correspond to networks with star topologies while efficient strategies correspond to complete networks in which every possible edge is present. Furthermore, by Corollary~\ref{cor:uc:star-max} we know that the upper bound on the inequality ratio for this regime of $\alpha$ is achieved in a star topology. Hence, to prove the second item of Theorem~\ref{thm:uc:ee} we need only address the existence of equilibrium strategies for which there is cost equality among the agents.

\begin{lemma}
\label{lem:uc:ee:2}
	When $1 \leq \alpha < 2$, the UC game admits a Nash equilibrium strategy that is egalitarian in the limit as $n \rightarrow \infty$.
\end{lemma}

\begin{proof}
Consider a joint strategy $s$ inducing a star network topology among the $|N| = n$ agents in which the central agent $c \in N$ buys $k = 1 + \lfloor \frac{n-2}{\alpha} \rfloor$ edges $\{c,i\}$ to agents $N_{\lnot b} \subset N$; all the remaining edges $\{j,c\}$ are paid for by their respective peripheral agent $j \in N_b = N \setminus N_{\lnot b} \cup \{c\}$. (The sets $N_b$ and $N_{\lnot b}$ partition $N$ so that all agents that buy edges appear in $N_b$ and all those who do not buy any edges appear in $N_{\lnot b}$.) In this strategy, the agents $i \in N_{\lnot b}$ all incur a cost $c_i(s) = 2n - 3$, all agents $j \in N_b$ incur a cost of $c_j(s) = 2n - 3 + \alpha$, and $c$ incurs a cost $c_c(s) = n - 1 + \alpha k \leq 2n - 3 + \alpha$.
Since $c_i(s) < c_c(s) \leq c_j(s)$, the inequality ratio is:
	\begin{equation}
	\label{eq:uc:ee:2:1}
	\ir(s) = \frac{c_j(s)}{c_i(s)} = \frac{2n - 3 + \alpha}{2n - 3} = 1 + \frac{\alpha}{2n - 3},
	\end{equation}
which approaches $1$ in the limit as $n \rightarrow \infty$. It is easy to see that $s$ is in fact a Nash equilibrium since no agent stands to reduce their cost by deleting an edge (since doing so would disconnect the network, prompting an infinite cost) and, at the cost of $\alpha \geq 1$ per edge, the addition of any edges will decrease the agent's usage cost only by $1$ (per edge), yielding no net decrease in the agent's total cost.
\end{proof}

Finally, the third item of Theorem~\ref{thm:uc:ee} addresses the regime when $\alpha$ is a constant greater than or equal to $2$. In this regime, we show that there exist some efficient Nash equilibrium strategies that are egalitarian, and other efficient Nash equilibrium strategies that achieve maximal inequality.

\begin{lemma}
\label{lem:uc:ee:3}
	When $2 \leq \alpha < \infty$, there exist Nash equilibrium strategies in the UC game which are both efficient and egalitarian in the limit as $n \rightarrow \infty$; and there are Nash equilibrium strategies that are both efficient and achieve maximal inequality.
\end{lemma}

\begin{proof}
Recall from Proposition~\ref{prop:fab2003:range} that, when $\alpha \geq 2$ the star is an efficient Nash equilibrium strategy. Thus, the lemma will follow if we can show that some star strategies can support maximal inequality while others can support equality.

To show the first part of the lemma (the existence of an egalitarian strategy that yields to a star network topology), we can invoke the strategy constructed in the proof of Lemma~\ref{lem:uc:ee:2}. Recall that this strategy calls for the central agent $c \in N$ to buy $1 + \lfloor \frac{n-2}{\alpha} \rfloor$ edges to peripheral agents $N_{\lnot b} \subset N$, and all remaining agents $j \in N \setminus N_{\lnot b} \cup \{c\}$ buy a single edge to $c$. The inequality ratio for this strategy is expressed by Equation~\ref{eq:uc:ee:2:1}, which approaches $1$ in the limit as $n \rightarrow \infty$, and is hence egalitarian.

The second part of the lemma (the existence of a star-topology-yielding strategy that achieves the inequality upper bound established in Theorem~\ref{thm:uc:ub}) follows immediately from Corollary~\ref{cor:uc:star-max}, which states that inequality in the UC game is maximized in networks with a star topology.
\end{proof}

\subsection{The Undirected Bounded Budget Connections Game}
\label{sec:ee:ubbc}

\begin{figure}[t]
	\centering
	\includegraphics{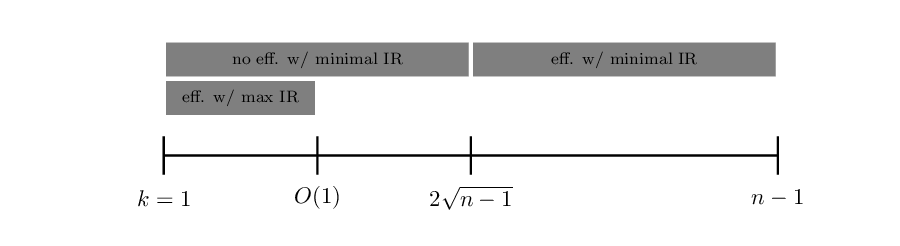}
	\caption{A summary of Theorem~\ref{thm:ubbc:ee}, characterizing the existence of Nash equilibrium strategies in the UBBC game with respect to efficiency and the inequality ratio for different uniform budgets, $k$.}
	\label{fig:ubbc:effeq}
\end{figure}

This section establishes the relationship between efficient Nash equilibrium strategies and (in)equality in the UBBC game with uniform budgets. We find that when the budget $k$ is sufficiently small ($k = O(1)$) we can construct an efficient Nash equilibrium strategy in which inequality is maximized; and when the edge budget is sufficiently large ($k \geq 2 \sqrt{n-1}$) then we show the existence of efficient Nash equilibrium strategies in which there is cost equality among all of the agents. Finally, we prove the non-existence of efficient Nash equilibrium strategies with egalitarian costs when $k < 2 \sqrt{n-1}$. These findings are formally presented in Theorem~\ref{thm:ubbc:ee} and summarized in Figure~\ref{fig:ubbc:effeq}.

\begin{theorem}
\label{thm:ubbc:ee}
The relationship between inequality, Nash equilibrium, and efficiency in the UBBC game with uniform budgets $k$:
	\begin{enumerate}
		\item \label{thm:ubbc:ee:1}
		When $k \geq 1$ is a constant (independent of $n$), there exists an efficient Nash equilibrium strategy that achieves the NIR upper bound of $2$ established in Theorem~\ref{thm:ubbc:ub}.
		
		\item \label{thm:ubbc:ee:2}
		When $k \geq 2 \sqrt{n-1}$, there exist efficient Nash equilibrium strategies that are egalitarian. 
		
		\item \label{thm:ubbc:ee:3}
		When $1 \leq k < 2 \sqrt{n-1}$, there does not exist an efficient Nash equilibrium strategy that is egalitarian.
	\end{enumerate}
\end{theorem}

\begin{lemma}
\label{lem:ubbc:ee:1}
In the UBBC game with constant (independent from $n$) uniform edge budgets $k \geq 1$, there exists an efficient Nash equilibrium strategy that achieves the NIR upper bound of $2$.
\end{lemma}

\begin{proof}
Fix $k$ and let $N = \{0, 1, \dots, n-1 \}$ denote the set of strategic agents. We will construct an efficient Nash equilibrium strategy $s$ as follows (see Figure~\ref{fig:ubbc:ee:1}):
	\begin{enumerate}
		\item $\forall i \in \{0, 1, \dots, k\}$, set $s_i = \{ i+1, i+2, \dots, i+k \}$.
		\item $\forall i \in \{ k+1, k+2, \dots, 2k-1 \}$, set $s_i = \{0, 1, \dots, i-k-1, i+1, i+2, \dots, 2k \}$.
		\item $\forall i \in \{ 2k, 2k+1, \dots, n-1 \}$, set $s_i = \{0, 1, \dots, k-1 \}$.
	\end{enumerate}
The strategy $s$ produces a network that has a diameter of $2$. Thus, per Lemma~\ref{lem:ubbc:diam2}, it is a Nash equilibrium.

To show that $s$ is socially efficient, we must determine its social cost. We can partition the agents $N$ into three sets that correspond to three ``tiers'' of costs: 
	\begin{itemize}
		\item Agents in $N_a = \{0, 1, \dots, k-1 \}$ are each directly connected to every other agent, so $c_i(s) = n - 1$ for all $i \in N_a$. Agents in $N_a$ incur the minimal cost among $N$.
		
		\item Agents in $N_b = \{2k+1, 2k+2, \dots, n-1 \}$ are each directly linked to every agent in $N_a$ and two-hops away from every other agent, so $c_i(s) = 2(n-1) - k$ for all $i \in N_b$. Agents in $N_b$ incur the maximal cost among $N$.
		
		\item Agents $N_c = \{k, k+1, \dots, 2k \}$ each have a direct link to every agent in $N_a$ and $N_c$ and a two-hop distance to agents in $N_b$, so $c_i(s) = 2(n - k - 1)$ for all $i \in N_c$.
	\end{itemize}
The cardinality of these partitions are $|N_a| = k$, $|N_b| = n - 2k - 1$, and $|N_c| = k+1$, and the social cost is
	\begin{eqnarray*}
	C(s) &=& \left[ k (n-1) \right] + \left[ (n - 2k - 1) (2 (n-1) - k) \right] + \left[ (k+1) (2 (n-k-1)) \right] \\
		&=& 2(n^2 - n - kn).
	\end{eqnarray*}
By Lemma~\ref{lem:ubbc:sc:eff}, this is the minimal social cost in a uniform game with size-$k$ budgets. Therefore, $s$ is socially efficient.

The inequality ratio is between a min-cost agent $i \in N_a$ and a max-cost agent $j \in N_b$;
	\begin{equation}
	\label{eq:ubbc:ee:1}
	\ir(s) = \frac{c_j(s)}{c_i(s)} = \frac{ 2(n-1) - k }{ n-1} = 2 - \frac{k}{n-1}.
	\end{equation}
Since \eqref{eq:ubbc:ee:1} approaches $2$ as $n \rightarrow \infty$, this strategy maximizes the inequality ratio established for the UBBC in Theorem~\ref{thm:ubbc:ub}. Hence, $s$ is an efficient Nash equilibrium strategy that achieves maximal inequality.
\end{proof}

\begin{figure}[t]
	\begin{minipage}[b]{.33\linewidth}
		\centering
		\includegraphics{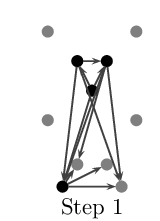}
	\end{minipage}
	\begin{minipage}[b]{.32\linewidth}
		\centering
		\includegraphics{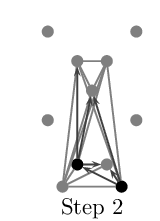}
	\end{minipage}
	\begin{minipage}[b]{.33\linewidth}
		\centering
		\includegraphics{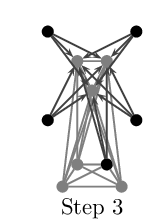}
	\end{minipage}
\caption{Example of the three parts to the strategy constructed in the proof of Lemma~\ref{lem:ubbc:ee:1}. Dark nodes and edges are those assigned in the specified step and arrows are meant to convey link ``ownership''.}
\label{fig:ubbc:ee:1}
\end{figure}

The next lemma establishes the second part of Theorem~\ref{thm:ubbc:ee}, showing that when edge budgets are sufficiently large, we can find an efficient Nash equilibrium strategy which achieves perfect equality among the agents. But first we need the following fact (see Appendix~\ref{app:proofs} for a proof):

\begin{fact}
\label{fact:diam2regular}
A $\delta$-regular graph of diameter two can have at most $n = \delta^2 + 1$ nodes.\footnote{Recall that a $\delta$-regular graph/network is a graph in which every node has $\delta$ neighbors.}
\end{fact}

\begin{lemma}
\label{lem:ubbc:ee:2}
In the UBBC game with uniform edge budgets $k \geq 2 \sqrt{n-1}$, there exists efficient Nash equilibrium strategies that are egalitarian. 
\end{lemma}

\begin{proof}
For a strategy $s$ to be efficient and egalitarian, it must be the case that $G_s$ has a diameter of at most $2$, and all agents have common one- and two-hop neighborhood sizes. When $k \geq (n-1)/2$, then this is trivially satisfied by the complete graph. Satisfying both of these properties for $k < (n-1)/2$ requires that $G_s$ be $2k$-regular with a diameter of $2$. From Fact~\ref{fact:diam2regular}, we know that such graphs are only possible if $k \geq 2 \sqrt{n-1}$.
\end{proof}

Finally, we turn to the third part of Theorem~\ref{thm:ubbc:ee}, which establishes the non-existence of efficient Nash equilibrium strategies that are egalitarian when edge budgets are below $2 \sqrt{n-1}$. 
The non-existence follows as a corollary of Fact~\ref{fact:diam2regular} and Proposition~\ref{prop:ubbc:diam2:eff} by observing that when $k < 2 \sqrt{n-1}$, there are not enough edges to create a regular, diameter-$2$ network.

\begin{corollary}
\label{cor:ubbc:ee:3}
When $1 \leq k < 2 \sqrt{n-1}$, there does not exist an efficient Nash equilibrium strategy that is egalitarian.
\end{corollary}

\section{Discussion and Conclusions}
\label{sec:con}
This paper examines inequality in two simple models of strategic network formation: the {\sc Undirected Connections} (UC) game of Fabrikant \etal \cite{Fab2003} and the {\sc Undirected Bounded Budget Connections} (UBBC) game introduced by Ehsani \etal \cite{Ehsani2011}. To this end, we introduce the Nash Inequality Ratio (NIR) as an instrument to quantify the level of inequality that can exist between a pair of agents in a Nash equilibrium. Upper bounds on the NIR are established for both games, and we show that these bounds are tight.

Our analysis of these games reveal an interesting relationship between scarcity (expressed by lower budgets in the UBBC model and high edge costs in the UC model) and inequality. In the UBBC game, we found that when edges are scarce the inequality upper-bound is attainable in Nash equilibrium; but when edges are more plentiful, the highest level of inequality is not sustainable in equilibrium. A similar correspondence is observed in the UC model, where higher edge costs can support maximal inequality in equilibrium. With respect to efficiency, we find that the two games behave differently from one another: in the UC game, scarcity can support efficient Nash equilibrium strategies with either egalitarian costs or maximal inequality; but the UBBC game only supports efficient Nash equilibrium strategies that not egalitarian when edges are scarce.

The NIR affords a wide-angle lens through which inequality can be analyzed, characterizing the \emph{extent} to which the costs incurred by a pair of agents can differ in equilibrium. It does not, however, provide a clear view on the \emph{distribution} of inequality among populations of agents, nor does it identify the inequality ratio of a ``typical'' Nash equilibrium strategy. Indeed, in our analysis of both network formation games we often relied on the inequality furnished by networks with a star topology in which there is a single, central agent that incurred the lowest (or sometimes highest) cost while the remaining $n-1$ agents all incurred the highest (or lowest) cost. In star networks, the inequality between a pair of agents selected uniformly at random is, with a high probability, nonexistent since neither agent will likely be the central agent. One interesting direction for future research that extends the analysis presented here to derive upper bounds on the expected inequality ratio between randomly selected pairs of agents. A related problem is to bound the inequality ratio of a randomly chosen Nash equilibrium strategy, or from the set of equilibrium strategies that result from a particular game dynamic like best- or approximate-best-response.

We believe that an analysis of the NIR will provide interesting insight into many other games beyond the two network formation games analyzed here. As with the Price of Anarchy, the NIR can be used to establish a bound on the ``price'' of strategic behavior by answering the question: To what extent can the speed of costs/benefits be found in equilibrium outcomes of distributed decision making by self-interested agents? It may be anarchy, but is it fair?

\subsection{Acknowledgments}
\label{sec:con:ack}
We gratefully acknowledge support from the U.S. Army Research Laboratory and the U. S. Army Research Office under Cooperative Agreement W911NF-09-2-0053 and MURI award W911NF-13-1-0340, as well as the Defense Threat Reduction Agency Basic Research Grant No. HDTRA1-10-1-0088. We would also like to thank the anonymous referees whose comments have helped improve this manuscript considerably.

\bibliographystyle{alpha} 
\bibliography{bib}

\appendix

\section{Omitted Proofs}
\label{app:proofs}

\begin{replemma}{lem:ubbc:sc:eff}
The social cost of any efficient strategy profile for a uniform UBBC instance with edge budgets $k \geq 1$ is
	\begin{equation}
	\label{eq:ubbc:sc:eff}
		2 n (n - 1) - 2 n k.
	\end{equation}
\end{replemma}

\begin{proof}
Let $C_{min}(n,k)$ be the social cost of an efficient outcome for a uniform UBBC instance with $n$ agents, each with a budget of $k$ edges. When $k = k_{max} = (n-1)/2$, the complete graph is the efficient outcome with a social cost of $C_{min}(n,k_{max}) = n(n-1)$.  Because an individual edge must lie on at least two shortest paths (\eg the edge $\{i,j\}$ is on both the $i \rightarrow j$ and $j \rightarrow i$ paths), starting from the complete network, every edge removal increases the social cost by at least $2$. Therefore, decreasing the edge budget by one reduces the social cost by at least $2n$. Hence, we can express $C_{min}(n,k)$ by
	$$C_{min}(n,k) = n(n-1) + 2n (k_{max} - k).$$
Substituting $k_{max} = (n-1)/2$ gives 
	$$C_{min}(n,k) = n(n-1) + 2n \left( \frac{n-1}{2} - k \right),$$
which can be rearranged into Equation~\eqref{eq:ubbc:sc:eff}, completing the proof.
\end{proof} 

\begin{repproposition}{prop:ubbc:diam2:eff}
Every uniform UBBC instance with $k < (n-1)/2$ has an efficient outcome with a diameter of $2$.
\end{repproposition}
Proving Propositon~\ref{prop:ubbc:diam2:eff} relies on the following lemma:

\begin{lemma}
\label{lem:ubbc:diam2:sc}
The social costs of all diameter-$2$ topologies with a fixed number of $m \leq \frac{n(n-1)}{2}$ non-overlapping edges are equivalent.
\end{lemma}

\begin{proof}
Suppose that $s$ is a joint strategy profile which induces a diameter-$2$ network $G_{s}$ with $m$ distinct (\ie non-parallel) edges. Because $G_{s}$ has a diameter of $2$, we can express the cost incurred by each agent $i \in N$ in terms of their degree and the cardinality of their one- and two-hop neighborhoods. Let $N_x(G_s, i)$ and $d(G_s,i)$ denote $i$'s $x$-hop neighborhood and degree, respectively, in $G_s$. The cost to $i$ can be expressed as
	\begin{eqnarray}
	c_i(s) 	&=& d(G_s,i) + 2 |N_2(G_s,i) \setminus N_1(G_s,i)| = d(G_si) + 2 \left(  |N_2(G_s,i)| - |N_1(G_s,i)| \right) \nonumber \\
			&=& d(G_s,i) + 2 \left( |N_2(G_s,i)| - d(G_s,i) - 1  \right) = 2 |N_2(G_s,i)| - d(G_si) - 2 \nonumber \\
			\label{eq:lem:diam2:cost}
			&=& 2n - d(G_s,i) - 2
	\end{eqnarray}
The substitution of $|N_2(i)|$ for $n$ in the last line is a requirement of our assumption that the network $G_{s}$ has a diameter of $2$. 

With the cost to an individual agent in a diameter-$2$ network established in Equation~\eqref{eq:lem:diam2:cost}, we can turn our attention to the social cost. Let $C(n,m)$ denote the social cost of a diameter-$2$ network $G$ with $n$ nodes and $m$ edges. 
	\begin{equation}
	\label{eq:lem:diam2:sc}
	C(n,m) = \sum_{i \in N} \left( 2n - d(G,i) - 2 \right) = 2n^2 - 2n - 2m.
	\end{equation}
The lemma follows from the fact that Equation~\eqref{eq:lem:diam2:sc} depends only on the diameter-$2$ assumption, the number of nodes, and the number of edges. 
\end{proof} 

\begin{proof}[Proof of Proposition~\ref{prop:ubbc:diam2:eff}]
From Lemma~\ref{lem:ubbc:sc:eff} we know that the social cost of an efficient outcome for a uniform UBBC instance is given by Equation~\eqref{eq:ubbc:sc:eff}, and from Lemma~\ref{lem:ubbc:diam2:sc} we know that the social cost of any diameter-$2$ network with $m$ edges is given by Equation~\eqref{eq:lem:diam2:sc}. A uniform UBBC instance with $k$ edges per agent induces a network with $m = nk$ edges. Substituting for $m$ and rearranging shows that both Equations \eqref{eq:ubbc:sc:eff} and \eqref{eq:lem:diam2:sc} are equivalent. The requirement that $k < (n-1)/2$ is a consequence of the fact that when $k \geq (n-1)/2$, the efficient outcome is the complete graph, which has a diameter of $1$.
\end{proof} 

\begin{replemma}{lem:fab2003:star_inequality}
As the number of agents grows toward infinity, the maximal inequality in a star topology for the UC network formation game with a constant edge cost $1 \leq \alpha < \infty$ is $\max \{2, (1 + \alpha) / 2\}$.
\end{replemma}

Note that the notation for the following proof of Lemma~\ref{lem:fab2003:star_inequality} differs from that used in the proof sketch given in Section~\ref{sec:upper_bounds:uc}. This is because the following proof must address more than just the center- and peripheral-sponsored stars as was done in the sketch.

\begin{proof}
Let $s$ be a strategy profile that produces a star topology rooted at agent $c$. Suppose that $k \in \{0, 1, \dots, n-1\}$ is the number of edges that $c$ purchases; \ie $|s_c| = k$. When $k=0$, we get the \emph{peripheral-sponsored star}, and when $k = n-1$, we get the \emph{center-sponsored star}. If $k < n-1$, then $\exists j \in N$ such that $c_j(s) = 2n - 3 + \alpha$, meaning that $j$ had to purchase the edge $\{j, c\}$. Similarly, when $k > 0$, then $\exists i \in N$ such that $c_i(s) = 2n - 3$, meaning that $c$ covered the cost of the edge $\{c, i\}$.

Partition the set of agents $N$ into three sets: 
	\begin{itemize}
	\item $N_c = \{c\}$ is the singleton consisting of the central agent,
	\item $N_{b} = \{i : c \in s_i\} \subseteq N \setminus \{c\}$ is the set of agents who built an edge to $c$, and 
	\item $N_{\lnot b} = \{j : j \in s_c \} \subseteq N \setminus \{c\}$ is the set of ``free-loading'' agents who do not buy an edge to $c$. 
	\end{itemize}
Note that we have $|N_{\lnot b}| = k$ and $|N_b| = n - k - 1$. 
All agents within a particular part are cost-equivalent, so the cost to an agent $i$ given the joint strategy profile $s$ is
	\begin{equation*}
	c_i(s) = \left\{	\begin{array}{rcl}
					n - 1 + \alpha k 	& \mbox{if} & i \in N_c \\
					2n - 3 + \alpha 	& \mbox{if} & i \in N_b \\
					2n - 3 		& \mbox{if} & i \in N_{\lnot b}
					\end{array} \right.
	\end{equation*}
	
When $k > (n-2) / \alpha + 1$, agents in $N_c$ incur the highest cost and agents in $N_{\lnot b}$ incur the lowest cost. Therefore, the inequality ratio is $(n - 1 + \alpha k) / (2n - 3)$. This ratio is maximized when $k = n-1$ (\ie in the center-sponsored star), so the maximal inequality ratio between agents in $N_c$ and $N_{\lnot b}$ is $( \alpha + 1)(n - 1) / ( 2n - 3)$, giving us
	\begin{equation} 
	\label{eq:fab2003:star_inequality:4}
	\lim_{n \rightarrow \infty} \ir = \frac{1 + \alpha}{2}.
	\end{equation}
	
When $k < (n-2)/\alpha$, the central agent in $N_c$ incurs the lowest cost while agents in $N_b$ incur the highest cost. In this case the inequality ratio is $(2n - 3 + \alpha) / (n - 1 + \alpha k)$. This quantity is maximized when $k = 0$ (\ie in the peripheral-sponsored star), so the maximal inequality ratio between agents in $N_c$ and $N_b$ is simply $( 2n - 3 + \alpha ) / ( n - 1)$, giving us
	\begin{equation}
	\label{eq:fab2003:star_inequality:6}
	\lim_{n \rightarrow \infty} \ir = 2.
	\end{equation}

Hence, the largest inequality ratio for a star topology is the maximum between Equations~\eqref{eq:fab2003:star_inequality:4} and \eqref{eq:fab2003:star_inequality:6} as $n \rightarrow \infty$.
\end{proof} 

\begin{repfact}{fact:diam2regular}
A $\delta$-regular graph of diameter two can have at most $n = \delta^2 + 1$ nodes.
\end{repfact}

\begin{proof} 
This fact follows from the Moore bound (\cf the survey by Miller and Sir\'{a}n \cite{Miller2013}). For completeness, we summarize the relevant part of the presentation in \cite[Section~3.1]{Miller2013} regarding the Moore bound.

Consider a node $v$ in a $\delta$-regular graph $G$. Let $n_i$ denote the number of nodes at distance $i$ from $v$ in $G$. We can bound $n_i \leq (\delta - 1) n_{i-1}$. If $G$ has a diameter $D$ then 
	\begin{eqnarray*}
	n_D = \sum_{i = 0}^D n_i &\leq& 1 + \delta + \delta(\delta - 1) + \cdots + \delta (\delta - 1)^D \\
		&=& 1 + \delta (1 + (\delta - 1) + \cdots + (\delta - 1)^{D-1}) \\
		&=& { \left\{ 	\begin{array}{rl}
					1 + \delta \frac{(\delta - 1)^D - 1}{\delta - 2} & \text{if } \delta > 2 \\
					2D + 1 & \text{if } \delta = 2
				\end{array} \right. }
	\end{eqnarray*}
Hence, when the diameter is $D = 2$, there can be at most $\delta^2 + 1$ nodes.
\end{proof} 

\end{document}